\documentclass[11pt]{article}

\newcommand{\Cov}{\operatorname{Cov}}

\usepackage{graphicx}
\usepackage{cite}
\usepackage{color}
\usepackage{amsfonts}
\usepackage{amsthm}
\usepackage{amsmath}
\usepackage{amssymb}
\usepackage{subfig}
\usepackage{fixltx2e}

\usepackage{pseudocode}
\usepackage[]{fullpage}

\newtheorem{theorem}{Theorem}
\newtheorem{lemma}{Lemma}
\newtheorem{remark}{Remark}
\newtheorem{definition}{Definition}
\newtheorem{corollary}{Corollary}

\newtheorem{proposition}{Proposition}
\newtheorem{example}{Example}
\newcommand{\EE}{\mathbb{E}}

\title{Entropy Jumps for Radially Symmetric Random Vectors}

\author{Thomas~A.~Courtade \\ Department of Electrical Engineering and Computer Sciences\\University of California, Berkeley}
\begin{document}

\maketitle
 \begin{abstract}
We establish a quantitative bound on the entropy jump associated to the sum of independent, identically distributed (IID) radially symmetric random vectors having dimension greater than one.  Following the usual approach, we first consider the analogous problem of Fisher information dissipation, and then integrate along the Ornstein-Uhlenbeck semigroup to obtain an entropic inequality.  In a departure from previous work, we appeal to a result by Desvillettes and  Villani on entropy production associated to the Landau equation.  This obviates strong regularity assumptions, such as presence of a spectral gap and log-concavity of densities, but comes at the expense of radial symmetry.    As an application, we give a quantitative estimate of the deficit in the Gaussian logarithmic Sobolev inequality for radially symmetric functions. 
    \end{abstract}

\section{Introduction}\label{sec:Intro}
Let $X$ be a random vector on $\mathbb{R}^d$ with density $f$.  The entropy associated to $X$ is defined by 
\begin{align}
h(X) = -\int_{\mathbb{R}^d} f \log f,
\end{align}
provided the integral exists. The non-Gaussianness of $X$, denoted by $D(X)$,  is given by
\begin{align}
D(X) = h(G^X) - h(X),
\end{align}
where $G^X$ denotes a Gaussian random vector with the same covariance as $X$.  Evidently, $D(X)$ is the relative entropy of $X$ with respect to $G^X$, and is therefore nonnegative. Moreover, $D(X)=0$ if and only if $X$ is Gaussian.

Our main result may be informally stated as follows:  Let $X, X_*$  be IID radially symmetric random vectors on $\mathbb{R}^d$,  $d\geq 2$, with sufficiently regular density $f$.  For any $\varepsilon > 0$
\begin{align}
h(\tfrac{1}{\sqrt{2}}(X+X_*) ) - h(X)    
   &\geq   C_{\varepsilon}(X) D(X) ^{1+\varepsilon},\label{mainBound}
\end{align}
where  $C_{\varepsilon}(X)$ is an explicit function depending  on $\varepsilon$,   the regularity of $f$, and a finite number of moments of $X$. In particular, if $\EE|X|^{4+\delta}<\infty$ for some $\delta>0$, then $C_{2/\delta}(X)>0$.  A precise statement can be found in Section \ref{sec:MainResults}, along with an analogous result for Fisher information and a related estimate that imposes no regularity conditions.  In interpreting the result, it is important to note that, although a radially symmetric density $f:\mathbb{R}^d \to [0,\infty)$ has a one-dimensional parameterization, the convolution $f \ast  f$ is inherently a $d$-dimensional operation unless $f$ is Gaussian.   Thus, it does not appear that  \eqref{mainBound} can be easily reduced to a one-dimensional problem. 

The quantity $h(\tfrac{1}{\sqrt{2}}(X+X_*) ) - h(X)$ characterizes the \emph{entropy production} (or, entropy jump) associated to $X$ under rescaled convolution.  Similarly, letting $J(X) = 4 \int_{\mathbb{R}^d}  \left|\nabla \sqrt f  \right|^2$ denote the Fisher information of $X$,  the quantity $J(X) - J(\tfrac{1}{\sqrt{2}}(X+X_*) )$ characterizes the \emph{dissipation of Fisher information}.  By the convolution inequalities of Shannon \cite{shannon48}, Blachman \cite{blachman1965convolution} and Stam \cite{stam1959some} for entropy power and Fisher information, it follows that both the production of entropy  and the dissipation of Fisher information  under rescaled convolution are nonnegative.  Moreover, both quantities are identically zero if and only if $X$ is Gaussian. 

The fundamental problem of bounding entropy production (and dissipation of Fisher information) has received considerable attention, yet quantitative bounds are few.  In particular, the entropy power inequality establishes that entropy production is strictly greater than zero unless $X$ is Gaussian, but gives no indication of how entropy production behaves as, say, a function of $D(X)$ when $X$ is non-Gaussian. As a consequence, basic stability properties of the entropy power inequality remain elusive, despite it being a central inequality in information theory.  For instance, a satisfactory answer to the following question is still out of reach: \emph{If the entropy power inequality is nearly saturated, are the random summands involved quantifiably close to  Gaussian?} 

 Perhaps the first major  result to address the question of entropy production in this setting is due to Carlen and Soffer \cite{carlen1991entropy}, who showed for each random vector $X$ with $J(X)<\infty$ and $\Cov(X)=\mathrm{I}$, there exists a nonnegative function $\Theta_{\psi,\chi}$ on $[0,\infty)$, strictly increasing from 0, and depending only on the two auxiliary functions 
 \begin{align}
 \psi(R) =   \EE[|X|^2 \mathbf{1}_{\{|x|\geq R\}}],
 \end{align}
 and 
 \begin{align}
 \chi(t) =   h\left(X_t \right)- h(X), ~~~~X_t :=e^{-t}X + (1-e^{-2t})^{1/2}G^X,
 \end{align}
such that 
\begin{align}
h(\tfrac{1}{\sqrt{2}}(X+X_*) ) - h(X) \geq \Theta_{\psi,\chi}(  D(X)  ).
\end{align}
Moreover, the function $ \Theta_{\psi_0,\chi_0}$ will do for any density $f$ with $\psi_f \leq \psi_0$ and $\chi_f \leq \chi_0$.  Hence, this provides a nonlinear estimate of entropy production in terms of  $D(X)$ that holds uniformly over all probability densities that exhibit the same decay and smoothness properties (appropriately defined).  Unfortunately, the proof establishing existence of $\Theta_{\psi,\chi}$ relies on a compactness argument, and therefore falls short of giving satisfactory quantitative bounds.  

A random vector $X$ with density $f$ has  spectral gap $c>0$ (equivalently, finite Poincar\'e constant $1/c$) if, for all smooth functions $g$ with $\int_{\mathbb{R}^d} fg = 0$,
\begin{align}
c \int_{\mathbb{R}^d} f g^2 \leq \int_{\mathbb{R}^d} f |\nabla g|^2.
\end{align}
 Generally speaking, a non-zero spectral gap is a very strong regularity condition on the density $f$ (e.g., it implies $X$ has finite moments of all orders).  
  
In dimension $d=1$, if $X$ has spectral gap $c>0$, then  \cite{ball2003entropy, johnson2004fisher} established the linear bound\footnote{Technically speaking, Barron and Johnson establish a slightly different inequality.  However, a  modification of their argument  gives the same result as Ball, Barthe and Naor.  See the discussion surrounding \cite[Theorem 2.4]{johnson2004information}.}
\begin{align}
h(\tfrac{1}{\sqrt{2}}(X+X_*) ) - h(X) \geq \tfrac{c}{2+2c} D(X).\label{specGap}
\end{align}
In dimension $d\geq 2$, Ball and Nguyen \cite{ball2012entropy} recently established an essentially identical result under the additional assumption that $X$ is isotropic (i.e., $\Cov(X)=\mathrm{I}$) with log-concave density $f$. Along these lines, Toscani has a strengthened EPI for log-concave densities \cite{toscani2015strengthened}, but the deficit is qualitative in nature in contrast to  the quantitative estimate obtained by Ball and Nguyen. %

Clearly,  entropy production and Fisher information dissipation is closely related to convergence rates in the entropic and Fisher information central limit theorems.  Generally speaking though, bounds in the spirit of \eqref{specGap} are unnecessarily strong for establishing entropic central  limit theorems of the form $D(S_n) = O(1/n)$, where $S_n := \tfrac{1}{\sqrt{n}}\sum_{i=1}^n X_i$ denotes the normalized sum of $n$ IID copies of $X$.  Indeed, it was long conjectured that $D(S_n) = O(1/n)$ under moment conditions. This was positively resolved by Bobkov, Chistyakov and G\"otze~\cite{bobkov2013rate, bobkov2014berry} using Edgeworth expansions and local limit theorems.  By Pinsker's inequality, we know that $D(X)$ dominates squared total-variation distance, so    $D(S_n) = O(1/n)$ is interpreted as a version of the Berry-Esseen theorem for the entropic CLT with the optimal convergence rate.  However, while the results of Bobkov et al. give good long-range estimates of the form $D(S_n) = O(1/n)$ (with explicit constants depending on the cumulants of $X$), the smaller-order terms propagate from local limit theorems for Edgeworth expansions and are non-explicit.  Thus,  explicit bounds  for the initial entropy jump $h(\tfrac{1}{\sqrt{2}}(X+X_*) ) - h(X)$ cannot be readily obtained.

Along these lines, we remark that   Ledoux, Nourdin and   Peccati  \cite{ledoux2016stein} recently established the weaker convergence rate $D(S_n) = O(\frac{\log n}{n})$ via the explicit bound 
\begin{align}
D(S_n) \leq \frac{S^2(X|G) }{2 n} \log\left( 1 + n \frac{I(X|G)}{S^2(X|G)}\right) , \label{steinDeficit}
\end{align}  
  where $I(X|G)$ is the relative Fisher information of $X$ with respect to the standard normal $G$, and $S(X|G)$ denotes the \emph{Stein discrepancy} of $X$ with respect to $G$, which is defined when a Stein kernel exists  (see \cite{ledoux2016stein} for definitions). In principle, this has potential to give an explicit bound on the entropy production  $h(\tfrac{1}{\sqrt{2}}(X+X_*) ) - h(X)$  in terms of $S^2(X|G)$  by considering $n=2$.
   Unfortunately, a Stein kernel may not always exist;  even if it does, further relationships between $S^2(X|G)$ and  $I(X|G)$ or $D(X)$ would need to be developed to ultimately yield a bound like \eqref{mainBound}.
  
Another related line of work  in statistical physics considers quantitative bounds on entropy production in the Boltzmann equation (see the review \cite{villani2002review} for an overview).  The two problems are not the same, but there is a strong analogy between entropy production in the Boltzmann equation and entropy jumps associated to rescaled convolution as can be seen by comparing  \cite{carlen1992strict} to \cite{carlen1991entropy}.  The details of this rich subject are  tangential to the  present discussion, but we  remark that a major milestone in this area was achieved when the entropy production in the Boltzmann equation was bounded from below by an explicit function of $D(X)$ (and various norms of $X$), where $X$ models the velocity of a particle in a rarified gas\cite{toscani1999sharp,villani2003cercignani}.  A key ingredient used to prove this bound was an earlier result by Desvillettes and Villani that controls relative Fisher information $I(X|G)$ via  entropy production in the  Landau equation:

\begin{lemma}\label{lem:DV} \cite{desvillettes2000spatially}
Let $X$ be a random vector on $\mathbb{R}^d$,  satisfying $\EE|X|^2 = d\geq 2$ and having density $f$.  Then, 
\begin{align}
\frac{1}{2} \iint  \left| x-x_*\right|^2f(x) f(x_*) \left|  \Pi(x-x_*)\left[  \frac{\nabla f}{f}(x) -\frac{\nabla f}{f}(x_*) \right] \right|^2 \mathrm{d}x \, \mathrm{d}x_*\geq \lambda \, (d-1)  I(X|G), \label{landauBound}
\end{align}
where $\lambda$ is the minimum eigenvalue of the covariance matrix associated to $X$, and $\Pi(v)$ is the orthogonal projection onto the subspace orthogonal to $v\in \mathbb{R}^d$. 
\end{lemma}

Our proof of \eqref{mainBound} follows a program similar to \cite{toscani1999sharp,villani2003cercignani}, and is conceptually straightforward after the correct ingredients are assembled.  In particular, we begin by recognizing that the LHS of \eqref{landauBound} resembles dissipation of Fisher information when written in the context of $L^2$ projections (cf. \cite[Lemma 3.1]{johnson2004fisher}).  Using the radial symmetry assumption, we are able to bound the Fisher information dissipation from below by  error terms plus entropy production in the Landau equation, which is subsequently bounded by relative Fisher information using Lemma \ref{lem:DV}.  Care must be exercised in order to control  error terms (this is where our regularity assumptions enter), but the final result \eqref{mainBound} closely parallels that proved in \cite{toscani1999sharp}  for the Boltzmann equation.  We remark that the assumption of a non-vanishing Boltzmann collision kernel in \cite{toscani1999sharp} has a symmetrizing effect on the particle density functions involved; the rough analog in the present paper is the radial symmetry assumption.

\subsection*{Organization}
The rest of this paper is organized as follows.  Section \ref{sec:Notation} briefly introduces notation and definitions that are used throughout.  Main results are stated and proved in Section \ref{sec:MainResults}, followed by a brief discussion on potential extensions to non-symmetric distributions.  Section \ref{sec:LSI} gives an application of the results to bounding the deficit in the Gaussian logarithmic Sobolev inequality.

\section{Notation and Definitions}\label{sec:Notation}

For a vector $v\in \mathbb{R}^d$, we let $|v| :=(\sum_{i=1}^d v_i^2  )^{1/2}$ denote its Euclidean norm.   For a random variable $X$ on $\mathbb{R}$ and $p\geq 1$, we write $\|X\|_p := \left( \EE |X|^p \right)^{1/p}$ for the usual $L^p$-norm of $X$.  It will be convenient to  use the same notation for $0 < p <1$, with the understanding that $\| \cdot \|_p$ is not a norm in this case.

Throughout,  $G\sim N(0,\mathrm{I})$ denotes a standard Gaussian random vector on $\mathbb{R}^d$; the dimension will be clear from context.  For a random vector $X$ on $\mathbb{R}^d$, we let $G^X = \sqrt{d^{-1}\EE|X|^2} G$ be a normalized Gaussian vector, so that $\EE|X|^2 = \EE |G^X|^2$.  For $d\geq 2$, we denote the coordinates of a random vector $X$ on $\mathbb{R}^d$ as $(X_1,X_2, \dots, X_d)$.  Thus, for example, $G_1^X$ is a zero-mean Gaussian random variable with variance $d^{-1}\EE|X|^2$.

For a random vector $X$ with smooth density\footnote{All densities are with respect to Lebesgue measure.} $f$, we define the Fisher information
\begin{align}
J(X) = 4 \int \left|\nabla \sqrt f  \right|^2 = \int_{f>0} \frac{|\nabla f|^2}{f} \label{FIidentity}
 \end{align}
 and the entropy 
\begin{align}
h(X) = -\int f \log f,
\end{align}
where `$\log$' denotes the natural logarithm throughout.  For random vectors $X,Q$ with respective densities $f,g$, the relative Fisher information is defined by
\begin{align}
I(X|Q) =   4 \int g \left|\nabla \sqrt {f/g}  \right|^2
\end{align}
and the relative entropy is defined by
\begin{align}
D(X|Q) = \int f \log \frac{f}{g}.%
\end{align}
Evidently, both quantities are nonnegative and 
\begin{align}
&I(X) := I(X|G^X) = J(X) - J(G^X)  &  D(X):=D(X|G^X) = h(G^X) -h(X).
\end{align}

Finally, we recall two basic inequalities that will be taken for granted  several times without explicit reference:  for real-valued $a,b$ we have  $(a+b)^2 \leq 2 a^2 +2 b^2$, and for random variables $X,Y$, we have Minkowski's inequality: $\|X+Y\|_p \leq \|X\|_p + \|Y\|_p$ when $p\geq 1$.

\begin{definition}
A random vector $X$ with density $f$ is radially symmetric if $f(x) = \phi(|x|)$ for some function $\phi : \mathbb{R}\to [0,\infty)$.
\end{definition}

We primarily concern ourselves with random vectors that satisfy certain mild regularity conditions.  In particular, it is sufficient to control  $|\nabla \log  f (x) |$ pointwise in terms of $|x|$. 
\begin{definition}
A random vector $X$ on $\mathbb{R}^d$ with  smooth density $f$ is $c$-regular if, for all $x\in \mathbb{R}^d$,
\begin{align}
|\nabla \log  f (x) | \leq c \left(  |x|+  \EE|X|  \right). \label{cRegIneq}
\end{align}
\end{definition}
We remark that the smoothness requirement of $f$ in the definition of $c$-regularity is stronger than generally required for our purposes.  However, it allows us to  avoid further qualifications; for instance, the identities \eqref{FIidentity} hold for any $c$-regular function.  Moreover, since $\nabla \log f = \tfrac{\nabla f}{f}$  for smooth $f$, we have $J(X)<\infty$ for any $c$-regular $X$ with $\EE|X|^2<\infty$.

Evidently, $c$-regularity  quantifies the smoothness of a density function.  The following important example  shows that any density can be   mollified to make it  $c$-regular.
\begin{proposition}\label{prop:addNoise} \cite{polyWuWasserstein2016}
Let $X$ and $Z$ be independent, where $Z\sim N(0,\sigma^2 \mathrm{I})$.  Then $Y = X+Z$ is $c$-regular with $c = 4  /\sigma^2$. 
\end{proposition}

Observe that, in the notation of the above proposition, if $X$ is radially symmetric then so is $Y$.  Therefore,  Proposition \ref{prop:addNoise} provides a convenient means to construct  radially symmetric random vectors that are $c$-regular.  Indeed, we have the following useful corollaries (proofs are found in the appendix).  

\begin{proposition}\label{prop:StartRegular}
Let $X$ be a random vector on $\mathbb{R}^d$, and let $X_t = e^{-t} X + (1-e^{-2t})^{1/2}G$ for  $t\geq0$.  If $X$ is $c$-regular, then $X_t$ is $(5c\,e^{2t})$-regular.
\end{proposition}

\begin{proposition}\label{prop:approxR}
Let $R_0$ be a non-negative random variable with $\EE R_0^2 = 1$ and distribution function $F_{R_0}$.  For any $d\geq 2$ and $t, \varepsilon>0$, there exists a $(4/\varepsilon)$-regular radially symmetric random vector $X$ on $\mathbb{R}^d$ with $\EE|X|^2 = d$ and $R = \frac{1}{\sqrt{d}}|X|$ satisfying 
\begin{align}
F_{R_0}\left(  \tfrac{ r - \sqrt{(t+1)\varepsilon}}{\sqrt{1-\varepsilon}}  \right)   - e^{-d t^2 /8}   \leq F_{R}(r) & \leq  F_{R_0}\left(  \tfrac{ r + \sqrt{(t+1)\varepsilon}}{\sqrt{1-\varepsilon}}  \right)  +  e^{-d t^2 /8},
\end{align}
where $F_{R}$ is the distribution function of $R$. 
\end{proposition}

\section{Main Results } \label{sec:MainResults}
In this section, we establish quantitative estimates on entropy production and Fisher information dissipation under rescaled convolution.  As can be expected, we begin with an inequality for Fisher information, and then obtain a corresponding entropy jump inequality by integrating  along the Ornstein-Uhlenbeck semigroup. 

\subsection{Dissipation of Fisher Information under Rescaled Convolution}

 \begin{theorem}\label{thm:FIjumps}
Let $X, X_*$ be IID radially symmetric random vectors on $\mathbb{R}^d$,  $d\geq 2$, with $c$-regular density $f$.  For any $\varepsilon > 0 $
\begin{align}
J(X)-J(\tfrac{1}{\sqrt{2}}(X+X_*)) &\geq K_{\varepsilon}(X)    I(X)^{1+\varepsilon}, \label{FIIjump}
\end{align}
where 
\begin{align}
K_{\varepsilon}(X) = 
\frac{(\varepsilon / 8)^{\varepsilon}  }{   \left( 8 (1+\varepsilon) \right)^{1+\varepsilon}     }  
\cdot
\frac{  \||X|^2\|_1^{1+\varepsilon} }{c^{{2 \varepsilon}}  \left\|  |X|^2 \right\|_{ 2+1/\varepsilon}^{1+2 \varepsilon}   } 
. \label{KpX_FI}
\end{align}
\end{theorem}
 
 \begin{remark}
 We have made no attempt to optimize the constant $K_{\varepsilon}(X)$.  %
 \end{remark}
 
 A few comments are in order.  First, we note that  inequality \eqref{FIIjump} is invariant to scaling $t: X \mapsto tX$ for $t>0$.  Indeed, if $X$ is $c$-regular, then a change of variables shows that $tX$ is $(c /t^2)$-regular. So, using homogeneity of the norms, we find that 
 \begin{align}
 K_{\varepsilon}(t X) = t^{2\varepsilon} K_{\varepsilon}(X).
 \end{align}
 Combined with the property that $t^2 J(t X) = J(X)$, we have  
 \begin{align}
 K_{\varepsilon}(tX)    I(tX|G^{tX})^{1+\varepsilon} = t^{-2} K_{\varepsilon}(X)I(X|G^{X})^{1+\varepsilon},
 \end{align}
 which has the same scaling behavior as the LHS of \eqref{FIIjump}.  That is, 
 \begin{align}
 J(tX)-J(\tfrac{1}{\sqrt{2}}(tX+tX_*)) = t^{-2}\left( J(X)-J(\tfrac{1}{\sqrt{2}}(X+X_*))\right).
 \end{align}
 
 Second, inequality \eqref{FIIjump} does not contain any terms that explicitly depend on dimension.  However, it is impossible to say that inequality \eqref{FIIjump} is dimension-free in the usual sense that both sides scale linearly in dimension when considering product distributions.  Indeed, the product of two identical radially symmetric densities is again radially symmetric if and only if the original densities were Gaussian themselves, which corresponds to the degenerate case when the dissipation of Fisher information is identically zero. However, inequality  \eqref{FIIjump} does exhibit  dimension-free behavior in the following sense:   Suppose for simplicity that $X$ is normalized so that $\EE|X|^2=d$.    Since $X$ is radially symmetric, it can be expressed as the product of independent random variables $X =\sqrt{d} R  U$, where $U$ is uniform on the $(d-1)$-dimensional sphere $\mathbb{S}^{d-1}$ and $R=\frac{1}{\sqrt{d}}|X|$ is a nonnegative real-valued random variable satisfying $\EE R^2 = 1$.   Now, by the log Sobolev inequality and Talagrand's   inequality, we have
\begin{align}
I(X|G) \geq 2 D(X|G) \geq  W_2^2 \Big( (\sqrt{d} R U) , G \Big)  = d \, W_2^2 \Big( R   , \tfrac{1}{\sqrt{d}} |G|  \Big) 
\end{align}
The equality follows since, for any vectors $v,w$ we have  $|v-w|^2 \geq \big| |v|-|w|\big|^2$.  However, this can be achieved with equality by the coupling $G = |G| U$.  Thus, we have
\begin{align}
K_{\varepsilon}(X)    I(X|G)^{1+\varepsilon} &\geq \frac{(\varepsilon / 8)^{\varepsilon} \, d^{1+2 \varepsilon}W_2^{2\varepsilon} \Big( R   , \tfrac{1}{\sqrt{d}} |G|  \Big)   }{c^{{2 \varepsilon}}  \left( 8 (1+\varepsilon) \right)^{1+\varepsilon} \, \left\|  |X|^2 \right\|_{ 2+1/\varepsilon}^{1+2 \varepsilon}   }   I(X|G)\\
&= \frac{(\varepsilon / 8)^{\varepsilon} \,  W_2^{2\varepsilon} \Big( R   , \tfrac{1}{\sqrt{d}} |G|  \Big)   }{c^{{2 \varepsilon}}  \left( 8 (1+\varepsilon) \right)^{1+\varepsilon} \, \left\| R \right\|_{ 2+1/\varepsilon}^{1+2 \varepsilon}   }   I(X|G).
\end{align}
Now, we note that $\EE R^2 = \EE (\tfrac{1}{\sqrt d} |G|)^2 = 1$, so we have a bound of the form
\begin{align}
J(X)-J(\tfrac{1}{\sqrt{2}}(X+X_*)) &\geq \widetilde{K}_{\varepsilon}(X)    I(X|G^X),
\end{align}
where the function  $\widetilde{K}_{\varepsilon}(X)$ is effectively dimension-free  in that it only depends on the (one-dimensional) quadratic Wasserstein distance between $R$ and  $\tfrac{1}{\sqrt{d}} |G|$.   For $d\to \infty$, the law of large numbers implies that $\tfrac{1}{\sqrt{d}} |G|\to 1$ a.s. Therefore, $ W_2^2 \Big( R   , \tfrac{1}{\sqrt{d}} |G|  \Big)$ behaves similarly to $\EE|R-1|^2 \geq \operatorname{Var}(R)$ high dimensions.  Indeed, by the triangle inequality applied to $W_2$, 
\begin{align}
\left| W_2 \Big( R   , \tfrac{1}{\sqrt{d}} |G|  \Big) - \|R-1\|_2  \right| \leq \left\|\tfrac{1}{\sqrt{d}} |G| -1\right\|_2  = O\left(\tfrac{1}{\sqrt{d}}\right). %
\end{align}
So, we see that \eqref{FIIjump} depends very weakly on $d$ when the marginal distribution of $|X|$ is preserved and dimension varies. 

One important question remains: As dimension $d\to\infty$, do there exist   random vectors $X$ on $\mathbb{R}^d$ with sufficient regularity for which the associated random variable $R$ is not necessarily  concentrated around $1$?  The answer to this is affirmative in the sense of Proposition \ref{prop:approxR}: we may approximate any distribution function $F_{R_0}$ to within arbitrary accuracy, at the (potential) expense of increasing the regularity parameter $c$.   
 
\begin{proof}[Proof of Theorem \ref{thm:FIjumps}]
As remarked above, inequality \eqref{FIIjump} is invariant to scaling.  Hence, there is no loss of generality in assuming that $X$ is normalized according to $\EE|X|^2 = d$.  Also, since $X$ is radially symmetric, $X-X_*$ is equal to $X+X_*$ in distribution, therefore we seek to  lower bound the quantity 
\begin{align}
J(X) - J(\tfrac{1}{\sqrt{2}}(X+X_*))= J(X) - 2 J(X-X_*).
\end{align}
Toward this end, define $W = X-X_*$, and denote its density by $f_W$.  By the projection property of the score function of sums of independent random variables,  the following identity holds (e.g., \cite[Lemma 3.4]{johnson2004information}):  
\begin{align}
2 \left( J(X) - 2 J(X-X_*)\right) &= \EE\left| 2 \rho_W(W) - \left( \rho(X) - \rho(X_*)\right) \right|^2,
\end{align}
where $\rho = \nabla \log f$ is the score function of $X$ and $\rho_W = \nabla \log f_W$ is the score function of $W$.  

For $v\in \mathbb{R}^d$, let $\Pi(v)$ denote the orthogonal projection onto the subspace  orthogonal to  $v$. Now, we have
\begin{align}
2 J(X) - 4 J(X-X_*) &= \EE\left| 2 \rho_W(W) - \left( \rho(X) - \rho(X_*)\right) \right|^2\\
&\geq \EE\left| 2 \Pi(W) \rho_W(W) - \Pi(X-X_*)\left( \rho(X) - \rho(X_*)\right) \right|^2\label{eq1symm} \\
&= \EE\left|  \Pi(X-X_*)\left( \rho(X) - \rho(X_*)\right) \right|^2.\label{eq2symm}
\end{align}
The inequality follows since $\Pi(w) = \Pi(x-x_*)$ by definition, and $|v| \geq |\Pi(w) v|$ since $\Pi(w)$ is an orthogonal projection.  The last equality follows since $\Pi(w) \rho_W(w) =0$ due to the fact that $\Pi(w) \nabla f_W(w)$ is the tangential gradient of $f_W$, which is identically zero due to radial symmetry of $f_W$. 

Next, for any $R>0$, use the inequality
\begin{align}
1 \geq \frac{|x-x_*|^2}{ R^2} -  \frac{|x-x_*|^2}{ R^2} \mathbf{1}_{\{|x-x_*|> R\}}
 \end{align}
 to  conclude that 
\begin{align}
2 J(X) - 4 J(X-X_*) &\geq \frac{1}{R^2} \EE\left[\left| X-X_*\right|^2\left|  \Pi(X-X_*)\left( \rho(X) - \rho(X_*)\right) \right|^2\right]\notag \\
&-\frac{1}{R^2} \EE\left[\left| X-X_*\right|^2\left|  \Pi(X-X_*)\left( \rho(X) - \rho(X_*)\right) \right|^2  \mathbf{1}_{\{|X-X_*|> R\}}\right].
\end{align}
We bound the second term first.  By $c$-regularity and the triangle inequality, we have
\begin{align}
\left|  \Pi(x-x_*)\left( \rho(x) - \rho(x_*)\right)\right| \leq \left|  \rho(x) - \rho(x_*)\right| \leq c(|x|+|x_*|)+2c\EE|X|. \end{align}
So,  noting 
the inclusion 
\begin{align}
\{|x-x_*|> R\} \supseteq \{|x|\geq R/2, |x_*|\leq |x|   \} \cup\{|x_*|\geq R/2, |x|\leq |x_*|   \},
\end{align} we have the pointwise inequality
\begin{align}
&\mathbf{1}_{\{|x-x_*|> R\}}  |x-x_*|^2 \left| \Pi(x-x_*)  \left( \rho(x) - \rho(x_*)\right)  \right|^2  \\
&\leq   \mathbf{1}_{\{|x|\geq R/2, |x_*|\leq |x|   \} \cup\{|x_*|\geq R/2, |x|\leq |x_*|   \} }  |x-x_*|^2   \left|   \rho(x) - \rho(x_*)  \right|^2\\
&\leq \mathbf{1}_{\{|x|\geq R/2     \} }4 |x|^2   \left(   2c|x|+2 c\EE|X| \right)^2 + \mathbf{1}_{\{|x_*|\geq R/2     \} }4 |x_*|^2   \left(   2c|x_*|+2 c\EE|X|  \right)^2.
\end{align}
Taking expectations and using the fact that $X,X_*$ are IID, we have for any conjugate exponents $p,q\geq 1$ and $\beta>0$, 
\begin{align}
&\EE\left[\left| X-X_*\right|^2\left|  \Pi(X-X_*)\left( \rho(X) - \rho(X_*)\right) \right|^2  \mathbf{1}_{\{|X-X_*|> R\}}\right]\notag\\
&\leq 16 \,c^2\, \EE\left[\left| X \right|^2\left(   |X| +  \EE|X| \right)^2  \mathbf{1}_{\{|X|\geq R/2\}}\right]\\
&\leq  32 \,  c^2 \EE\left[\left| X \right|^4  \mathbf{1}_{\{|X|> R/2\}}\right] + 
           32 \,  c^2 \, \left(\EE|X|\right)^2 \,  \EE\left[\left| X \right|^2  \mathbf{1}_{\{|X|\geq R/2\}}\right]\\
&\leq  32 \,  c^2   \,\left\| |X|^2 \right\|_{ 2p}^2\left( \Pr\left\{ |X|\geq R/2\right\}\right)^{1/q}  
        + 32 \,  c^2 \, \left(\EE|X|\right)^2\,  \left\| |X|^2 \right\|_{ p} \left( \Pr\left\{ |X|\geq R/2\right\}\right)^{1/q}\\
&\leq  32 \,  c^2   \left\| |X|^2 \right\|_{ 2p}^2  \left( \frac{\EE|X|^{\beta}}{(R/2)^{\beta}}\right)^{1/q}
   + 32 \,  c^2 \left(\EE|X|\right)^2\,    \left\| |X|^2 \right\|_{ p}\left( \frac{\EE|X|^{\beta}}{(R/2)^{\beta}}\right)^{1/q}\\
&=  \frac{32\cdot 2^{\beta/q} \,  c^2}{R^{\beta /q}} \left(    \left\|  |X|^2 \right\|^2_{ 2p}  
   +   \left\| |X|^2 \right\|_{1/2}\,   \left\| |X|^2 \right\|_{ p}  \right) \left(\EE|X|^{\beta}\right)^{1/q}\\
&\leq  \frac{64\cdot 2^{\beta/q} \,  c^2}{R^{\beta /q}}   \left\|  |X|^2 \right\|^2_{ 2p}    \left\||X|^{2}\right\|_{\beta/2}^{\beta/(2q)} . 
\end{align}
Since $\EE|X|^2 = d$, radial symmetry implies $\Cov(X) = \mathrm{I}$.  Therefore, by Lemma \ref{lem:DV}, we have
\begin{align}
 \EE\left[\left| X-X_*\right|^2\left|  \Pi(X-X_*)\left( \rho(X) - \rho(X_*)\right) \right|^2\right] \geq 2(d-1) I(X|G).
\end{align}
Continuing from above, we have proved that
\begin{align}
J(X) - 2 J(X-X_*)&\geq \frac{d-1  }{R^2} I(X|G) -\frac{32\cdot 2^{\beta /q} \,  c^2 }{R^{2+\beta/q}}  \left\|  |X|^2 \right\|^2_{ 2p}    \left\||X|^{2}\right\|_{\beta/2}^{\beta/(2q)}. 
\end{align}
For any $s > 0$, Taking $R = \left( \frac{b(2+s)}{2a}\right)^{1/s}$ yields the identity
\begin{align}
\frac{a}{R^2}-\frac{b}{R^{2+s}} = \frac{1}{1+2/s}\left( \frac{2/s}{b(1+2/s)}\right)^{2/s} a^{1+2/s}.
\end{align}
So, putting  $\varepsilon = 2q/\beta$, $p = 1+1/(2\varepsilon)$ and simplifying, we obtain 
\begin{align}
J(X) - 2 J(X-X_*)&\geq \frac{(\varepsilon /8)^{\varepsilon}}{\left( 4 (1+\varepsilon) \right)^{1+\varepsilon}}\left(\frac{ 1 }{   c \left\|  |X|^2 \right\|_{ 2p}     }\right)^{2 \varepsilon} \left( \frac{1}{ \left\||X|^{2}\right\|_{q/\varepsilon}}\right)\left( (d-1) I(X|G)\right) ^{1+\varepsilon}\\
&\geq \frac{(\varepsilon /8)^{\varepsilon}}{\left( 8 (1+\varepsilon) \right)^{1+\varepsilon}}\left(\frac{ 1 }{   c \left\|  |X|^2 \right\|_{ 2p}     }\right)^{2 \varepsilon} \left( \frac{1}{ \left\||X|^{2}\right\|_{q/\varepsilon}}\right)\left( d\,  I(X|G)\right) ^{1+\varepsilon}\\
&= \frac{(\varepsilon / 8)^{\varepsilon}}{c^{2\varepsilon }\left( 8 (1+\varepsilon) \right)^{1+\varepsilon}}\left(\frac{ 1 }{   \left\|  |X|^2 \right\|_{ 2+\varepsilon^{-1}}     }\right)^{1+2 \varepsilon}  \left( d \, I(X|G)\right) ^{1+\varepsilon}\\
&= \frac{(\varepsilon / 8)^{\varepsilon} \, \||X|^2\|_1^{1+\varepsilon} }{c^{{2 \varepsilon}}  \left( 8 (1+\varepsilon) \right)^{1+\varepsilon} \, \left\|  |X|^2 \right\|_{ 2+1/\varepsilon}^{1+2 \varepsilon}   } \, I(X|G)  ^{1+\varepsilon},\label{FI_jump_beforeSimplify}
\end{align}
where we have made use of the crude bound $(d-1)/d\geq 1/2$ and substituted $d = \EE|X|^2 = \| |X|^2\|_1$.
\end{proof}

\subsection{Entropy Production under Rescaled Convolution }

As one would expect, we may `integrate up' in Theorem \ref{thm:FIjumps} to obtain an entropic version.  A precise version of the result  stated in Section \ref{sec:Intro} is given as follows:

\begin{theorem}\label{thm:Hjump}
Let $X, X_*$ be IID radially symmetric random vectors on $\mathbb{R}^d$,  $d\geq 2$, with $c$-regular density $f$.  For any $\varepsilon > 0 $
\begin{align}
h(\tfrac{1}{\sqrt{2}}(X+X_*))-h(X)    
   &\geq   C_{\varepsilon}(X)  \,D(X)^{1+\varepsilon} ,\label{EJineq}
\end{align}
where
\begin{align}
C_{\varepsilon}(X) =  \left( \frac{d\varepsilon}{1+( d+2)\varepsilon}\right)^{1+2\varepsilon}
       \frac{  2^4 \, (d/100)^{\varepsilon}  }
   {   
   \left( 2^8 (1+\varepsilon) (1+2\varepsilon) \right)^{1+\varepsilon} 
   } \cdot \frac{   \||X|^2\|_1  }{ c^{2\varepsilon}\, \left\|  |X|^2 \right\|_{ 2+1/\varepsilon}^{1+2 \varepsilon}  } .
  \end{align}
\end{theorem}

 \begin{remark}
 Although the constant $C_{\varepsilon}(X)$ appears to grow favorably with dimension $d$, this dimension-dependent growth can cancel to give a bound that is effectively dimension-free.  An illustrative example follows the proof. 
 \end{remark}
 
\begin{proof}
Similar to  before, the inequality \eqref{EJineq} is scale-invariant.  Indeed, all relative entropy terms are invariant to scaling $t  : X\mapsto t X$, and we also have $C_{\varepsilon}(tX) =C_{\varepsilon}(X) $ due to $tX$ being $(c/t^2)$-regular if $X$ is $c$-regular and homogeneity of the norms.  Thus, we may assume without loss of  generality that  $X$ is normalized so that $\EE|X|^2 = d$.   
Next, define $W = \tfrac{1}{\sqrt{2}}(X+X_*)$, and let $X_t, W_t$ denote the Ornstein-Uhlenbeck evolutes of $X$ and $W$, respectively. That is, for $t\geq 0$
\begin{align}
&X_t = e^{-t} X + (1-e^{-2t})^{1/2}G, & W_t = e^{-t} W + (1-e^{-2t})^{1/2}G.
\end{align}

By Proposition \ref{prop:StartRegular}, $X_t$ is $(5 c e^{2t})$-regular for all $t\geq 0$. Noting that $\EE|X_t|^2  = \EE|X|^2$, an application of Theorem \ref{thm:FIjumps}  gives 
\begin{align}
 I(X_t|G) - I(W_t|G)  &\geq \frac{(\varepsilon / 8)^{\varepsilon} \, \||X_t|^2\|_1^{1+\varepsilon} }{(5c)^{{2 \varepsilon}}  \left( 8 (1+\varepsilon) \right)^{1+\varepsilon} \, \left\|  |X_t|^2 \right\|_{ 2+1/\varepsilon}^{1+2 \varepsilon}   } \,e^{-4 \varepsilon t} I(X_t|G)  ^{1+\varepsilon}  \\
   &\geq \frac{(\varepsilon / 8)^{\varepsilon} \, \||X|^2\|_1^{1+\varepsilon} }
   {   (5 c)^{{2 \varepsilon}}  
   \left( 8 (1+\varepsilon) \right)^{1+\varepsilon} \,
   (2(\tfrac{1+(2+d)\varepsilon}{d\varepsilon}  ))^{1+2\varepsilon} 
   \left\|  |X|^2 \right\|_{ 2+1/\varepsilon}^{1+2 \varepsilon}  } 
   \,e^{-4 \varepsilon t} I(X_t|G)  ^{1+\varepsilon}  \label{ineqToJustify} \\
      &=
      \left( \frac{d\varepsilon}{2+(2d+4)\varepsilon}\right)^{1+2\varepsilon}
       \frac{   (\varepsilon / 8)^{\varepsilon} \, \||X|^2\|_1^{1+\varepsilon} }
   {   (5 c)^{{2 \varepsilon}}  
   \left( 8 (1+\varepsilon) \right)^{1+\varepsilon} \,
   \left\|  |X|^2 \right\|_{ 2+1/\varepsilon}^{1+2 \varepsilon}  } 
   \,e^{-4 \varepsilon t} I(X_t|G)  ^{1+\varepsilon}   ,\notag
\end{align}
where  \eqref{ineqToJustify} holds since, for $p\geq 1$,
\begin{align}
 \left\|  |X_t|^2 \right\|_{p}  = \left( \EE |X_t|^{2 \cdot p}  \right)^{1/p} &\leq 2 \left(  \EE (e^{-2t}|X|^2 +(1-e^{-2t})  |G|^2 )^{p}  \right)^{1/p}\\
 &= 2 \left\| e^{-2t}|X|^2 +(1-e^{-2t})  |G|^2 \right\|_{p}\\
  &\leq  2\left(   e^{-2t} \left\| |X|^2 \right\| _{p}+  (1-e^{-2t}) \left\|  |G|^2 \right\|_{p}  \right)\\
    &\leq  2 (1+ \tfrac{p}{d}) \left\| |X|^2 \right\| _{p} . \label{chiMoment}
\end{align}
The bound \eqref{chiMoment} uses the fact that $|G|^2$ is a chi-squared random variable with $d$ degrees of freedom,  and hence (using $\EE|X|^2 = d$):
\begin{align}
\left\|  |G|^2 \right\|_{p} = \left( 2^p \frac{\Gamma(p+\tfrac{d}{2})}{\Gamma(\tfrac{d}{2})} \right)^{1/p} &=\EE|X|^2 \left(  \frac{\Gamma(p+\tfrac{d}{2})}{\Gamma(\tfrac{d}{2})   \left( \tfrac{d}{2}\right)^p} \right)^{1/p}\\
 &\leq \EE|X|^2 \left( 1+\tfrac{p}{d} \right) \label{chiSquareBound} \\
 &\leq  \left\| |X|^2 \right\| _{p}(1+ \tfrac{p}{d}) .
\end{align}

Now, the claim will follow by integrating both sides.  Indeed, by the classical de Bruijn  identity, we have
\begin{align}
\int_{0}^{\infty}\left( I(X_t|G) - I(W_t|G)  \right) dt = D(X|G)- D(W|G) = h(\tfrac{1}{\sqrt{2}}(X+X_*))-h(X)  .
\end{align}
By Jensen's inequality, 
\begin{align}
\int_{0}^{\infty}e^{-4\varepsilon t }  I(X_t|G)^{1+\varepsilon}   dt  &\geq 
\frac{1}{(4\varepsilon)^{\varepsilon}}\left( \int_{0}^{\infty}e^{-4\varepsilon t }  I(X_t|G) dt \right) ^{1+\varepsilon}  \\
&\geq   \frac{1}{(4\varepsilon)^{\varepsilon}}  \left(  \int_{0}^{\infty}  I(X_{t+2\varepsilon t }|G) dt  \right) ^{1+\varepsilon} \\
&= \frac{1}{(4\varepsilon)^{\varepsilon}(1+2\varepsilon)^{1+\varepsilon}}D(X|G)^{1+\varepsilon} ,
\end{align}
where we used the bound $I(X_{t+s}|G)\leq e^{-  2 s}I(X_t|G)$ due to exponential decay of information along the semigroup (e.g., \cite{bakry2013analysis}), a change of variables, and the identity $\int_{0}^{\infty}I(X_t|G)   dt = D(X|G)$.
Thus, we have proved
\begin{align}
&h(\tfrac{1}{\sqrt{2}}(X+X_*))-h(X) \notag   \\
&\geq      
      \left( \frac{d\varepsilon}{2+(2d+4)\varepsilon}\right)^{1+2\varepsilon}
       \frac{   (\varepsilon / 8)^{\varepsilon} \, \||X|^2\|_1^{1+\varepsilon} }
   {   (5 c)^{{2 \varepsilon}}  
   \left( 8 (1+\varepsilon) \right)^{1+\varepsilon} \,
   \left\|  |X|^2 \right\|_{ 2+1/\varepsilon}^{1+2 \varepsilon}  } 
      \cdot   \frac{1}{(4\varepsilon)^{\varepsilon}(1+2\varepsilon)^{1+\varepsilon}}D(X|G)^{1+\varepsilon}   \\
&=    \left( \frac{d\varepsilon}{1+( d+2)\varepsilon}\right)^{1+2\varepsilon}
       \frac{  2^4 \, (d/100)^{\varepsilon}  }
   {   
   \left( 2^8 (1+\varepsilon) (1+2\varepsilon) \right)^{1+\varepsilon} 
   } \cdot \frac{   \||X|^2\|_1  }{ c^{2\varepsilon}\, \left\|  |X|^2 \right\|_{ 2+1/\varepsilon}^{1+2 \varepsilon}  }
    D(X|G)^{1+\varepsilon}  .
\end{align}

\end{proof}

\begin{example}[Centered Gaussian Mixtures]
Define $f_i$ to  be the density associated to the centered Gaussian distribution with covariance $\sigma_i^2\, \mathrm{I}$.  Let $X$ be a random vector on $\mathbb{R}^d$ with density $f = \sum_{i=1}^n p_i f_i$.  For convenience, assume $\sigma_1^2 =\min_i \sigma_i^2$ and that the $\sigma_i$'s are normalized so that $\sum_{i=1}^n p_i \sigma_i^2=1$ ($\{p_1, \dots, p_n\}$ is a probability vector).  Then, for $d$ large,
\begin{align}
h(\tfrac{1}{\sqrt{2}}(X+X_*))-h(X)    \gtrsim         \frac{  1  }
   {   
   2^{15}
   } \cdot \frac{    \sum_{i=1}^n p_i (\sigma_i-1)^2     }{  \sigma_1^2 \,  \sum_{i=1}^n p_i \left( \sigma^2_i/ \sigma_1^2\right)^{3}   }\,
    D(X).
\end{align}
\end{example}
It is easy to verify that $f$ is $(1/  \sigma_1^2  )$-regular.  Moreover, $\EE|X|^2 = d$, so using the bound \eqref{chiSquareBound}, we have
\begin{align}
\| |X|^2\|_{2+1/\varepsilon}\leq \left(\sum_{i=1}^n p_i \sigma_i^{2(2+1/\varepsilon)}  \right)^{1/(2+1/\varepsilon)}  (d + 2+1/\varepsilon) 
\end{align}
Now, by Talagrand's inequality, we may lower bound $D(X) \geq  \frac{d}{2} W_2^2 \left( \tfrac{1}{\sqrt{d}} |X| ,  \tfrac{1}{\sqrt{d}} |G|   \right)$ as we did in the discussion following Theorem \ref{thm:FIjumps}. 
Putting everything together and simplifying, we obtain:
\begin{align}
&h(\tfrac{1}{\sqrt{2}}(X+X_*))-h(X)       \\
&\geq    \left( \frac{d\varepsilon}{1+( d+2)\varepsilon}\right)^{2(1+2\varepsilon)}
       \frac{  2^4 \, (1/200)^{\varepsilon}  }
   {   
   \left( 2^8 (1+\varepsilon) (1+2\varepsilon) \right)^{1+\varepsilon} 
   } \cdot \frac{    W_2^{2\varepsilon } \left(  \tfrac{1}{\sqrt{d}} |X| ,  \tfrac{1}{\sqrt{d}} |G|   \right)     }{  \sigma_1^2 \, \left(\sum_{i=1}^n p_i \left( \sigma^2_i/ \sigma_1^2\right)^{2+1/\varepsilon}  \right)^{\varepsilon}   }\,
    D(X).\notag
\end{align}
An easy consequence of the LLN is the limit
\begin{align}
W_2^{2 } \left(  \tfrac{1}{\sqrt{d}} |X| ,  \tfrac{1}{\sqrt{d}} |G|   \right)  \to \sum_{i=1}^n p_i (\sigma_i-1)^2~~~\text{as $d\to\infty$},
\end{align}
so the claim follows by putting $\varepsilon=1$ and crudely bounding. 

 \medskip
 
It is straightforward to remove the explicit requirement in Theorem \ref{thm:Hjump} for $c$-regularity:

\begin{corollary}\label{cor:HjumpNoRegularity}
Let $X, X_*$ be IID radially symmetric random vectors on $\mathbb{R}^d$,  $d\geq 2$, with finite Fisher information.  For any $\varepsilon > 0 $
\begin{align}
h(\tfrac{1}{\sqrt{2}}(X+X_*))-h(X)   
 \geq   \widetilde{C}_{\varepsilon}(X)
    \frac{D(X)^{1+3\varepsilon}}{I(X)^{2\varepsilon} }   ,\label{noC_Hjumps}
\end{align}
where
\begin{align}
 \widetilde{C}_{\varepsilon}(X) =  \left( \frac{d\varepsilon}{1+( d+2)\varepsilon}\right)^{2+4\varepsilon}
       \frac{  2^{12} \, (d/100)^{\varepsilon}  }
   {   
   \left( 2^{17} (1+\varepsilon) (1+2\varepsilon) \right)^{1+\varepsilon} 
   } \cdot \frac{   \||X|^2\|_1  }{    %
   \left\|  |X|^2 \right\|_{ 2+1/\varepsilon}^{1+2 \varepsilon}   }.\label{CtildeDefn}
  \end{align}
\end{corollary}
\begin{remark}
Although the requirement of $c$-regularity is eliminated in \eqref{noC_Hjumps}, we see that the statement  of \eqref{noC_Hjumps}  is effectively  the same as  \eqref{EJineq}, with $c$-regularity  being replaced by another measure of regularity of $X$, i.e., the relative Fisher information $I(X)$.  
\end{remark}

\begin{proof}
Observe that inequality \eqref{noC_Hjumps} is invariant to scaling $X$.  Indeed, for $t>0$ we have $t^{4\varepsilon} \widetilde{C}_{\varepsilon}(t X) = \widetilde{C}_{\varepsilon}(X)$ and $t^2 I(tX | G^{tX})  =  I(X | G^{X})$.  Therefore, 
\begin{align}
\widetilde{C}_{\varepsilon}(t X) \frac{1}{I(tX|G^{tX})^{2\varepsilon} } =\widetilde{C}_{\varepsilon}(X) \frac{1}{I(X|G^{X})^{2\varepsilon} }. 
\end{align}
Since $D(tX | G^{tX})=D(X|G^X)$ and  $h(\tfrac{1}{\sqrt{2}}(tX+tX_*))-h(tX) =h(\tfrac{1}{\sqrt{2}}(X+X_*))-h(X)$, the claim follows.  Thus, as before, we will assume without loss of generality that  $X$ is normalized so that $\EE|X|^2 = d$.   

Next, define $W = \tfrac{1}{\sqrt{2}}(X+X_*)$, and let $X_t, W_t$ denote the Ornstein-Uhlenbeck evolutes of $X$ and $W$, respectively. That is, for $t\geq 0$
\begin{align}
&X_t = e^{-t} X + (1-e^{-2t})^{1/2}G, & W_t = e^{-t} W + (1-e^{-2t})^{1/2}G.
\end{align}

By Proposition \ref{prop:addNoise}, $X_t$ is $4\,(1-e^{-2t})^{-1}$-regular for all $t\geq 0$. Noting that $\EE|X_t|^2  = \EE|X|^2$, an application of Theorem \ref{thm:Hjump}  gives 
\begin{align}
&h(W_t) - h(X_t) \notag \\
 &\geq    \left( \frac{d\varepsilon}{1+( d+2)\varepsilon}\right)^{1+2\varepsilon}
       \frac{  2^8 \, (d/100)^{\varepsilon}  }
   {   
  \left( 2^{12} (1+\varepsilon) (1+2\varepsilon) \right)^{1+\varepsilon} 
   } \cdot \frac{(1-e^{-2t})^{2\varepsilon}   \||X_t|^2\|_1  }{  \left\|  |X_t|^2 \right\|_{ 2+1/\varepsilon}^{1+2 \varepsilon}  }
    D(X_t|G)^{1+\varepsilon} \\
    &\geq \left( \frac{d\varepsilon}{1+( d+2)\varepsilon}\right)^{1+2\varepsilon}
       \frac{  2^8 \, (d/100)^{\varepsilon}  }
   {   
   \left( 2^{12} (1+\varepsilon) (1+2\varepsilon) \right)^{1+\varepsilon} 
   } \cdot \frac{(1-e^{-2t})^{2\varepsilon}   \||X|^2\|_1  }{  (2(\tfrac{1+(2+d)\varepsilon}{d\varepsilon}  ))^{1+2\varepsilon} 
   \left\|  |X|^2 \right\|_{ 2+1/\varepsilon}^{1+2 \varepsilon}   }
    D(X_t|G)^{1+\varepsilon}  \label{ineqToJustify2}  \\
&= \left( \frac{d\varepsilon}{1+( d+2)\varepsilon}\right)^{2+4\varepsilon}
       \frac{  2^9 \, (d/100)^{\varepsilon}  }
   {   
   \left( 2^{14} (1+\varepsilon) (1+2\varepsilon) \right)^{1+\varepsilon} 
   } \cdot \frac{(1-e^{-2t})^{2\varepsilon}   \||X|^2\|_1  }{     %
   \left\|  |X|^2 \right\|_{ 2+1/\varepsilon}^{1+2 \varepsilon}   }
    D(X_t|G)^{1+\varepsilon}  ,
\end{align}
where  \eqref{ineqToJustify2} follows by the same logic as \eqref{ineqToJustify}.
Now, by de Bruijn's  identity and the convolution inequality for Fisher information, it follows that 
\begin{align}
\frac{d}{dt}\left( h(W_t) - h(X_t)\right) = J(W_t)-J(X_t) \leq 0.
\end{align}
Thus, $ h(W) - h(X)\geq  h(W_t) - h(X_t)$.  The map $t \mapsto D(X_t |G)$ is continuous on $t\in[0,\infty)$ (e.g., \cite{carlen1991entropy}).  Hence, using the fact that $\frac{d}{dt} D(X_t |G) = -I(X_t|G)$ (de Bruijn's identity) and that $t\mapsto I(X_t|G)$ is monotone-decreasing on $t\in[0,\infty)$ (convolution inequality for Fisher information), we have the inequality $
D(X_t |G) \geq D(X|G) - t I(X|G)$. 

Thus, to finish the proof, put $t = \frac{D(X|G)}{2 I(X|G)}$ and note that $1-e^{-x}\geq\tfrac{1}{\sqrt{2}}x$ for $x\in [0,1/4]$, which applies for our choice of $t$  due to the log Sobolev inequality, i.e.,  $\tfrac{1}{2}I(X|G)\geq D(X|G)$.
\end{proof}

\subsection{Beyond radial symmetry}

Theorem \ref{thm:Hjump} can be immediately extended to distributions that are radially symmetric, modulo an affine transformation.  Indeed, in this case, we can  apply the appropriate linear transformation $X \mapsto A (X - \mu)$, and then invoke the  formulae for the behavior of entropy under such transformations to bound the desired entropy jump.    A similar statement may be made for Theorem \ref{thm:FIjumps}. 

Less immediately, there is potential to  bound the entropy jumps associated to  a general non-symmetric random vector $X$ by considering its  symmetric decreasing rearrangement $X^{\star}$ (see, e.g., \cite{wang2014beyond} for definition).  In this case,  it is known that $h(X) = h(X^{\star})$ and $h(X+X_* )  \geq  h( X^{\star}+X^{\star}_*)$ for $X,X_*$ independent \cite{wang2014beyond}.  Thus, we might expect to be able to bound the entropy jump $h(\tfrac{1}{\sqrt{2}}(X+X_*) ) - h(X)$ from below by $D(X^{\star})$.  There are two issues here that must be  considered in order to proceed along this route.  The first issue is fundamental: we cannot expect to bound $D(X^{\star})$ by $D(X)$ in general (e.g., consider $X$ to be a rearrangement of $G$).  The second issue is technical and arises from the assumed regularity conditions.  In particular, one would need to establish conditions under which the regularity of $X^{\star}$ can be appropriately controlled by the regularity of $X$ (either in the sense of $c$-regularity or relative Fisher information). 

In any case, it may be possible to apply the general proof idea to non-symmetric $X$.  In particular, radial symmetry of $X$ was only critically used in passing from \eqref{eq1symm} to \eqref{eq2symm}.  In general, by the projection property of the score function, we have 
\begin{align}
&\EE\left| 2 \Pi(W) \rho_W(W) - \Pi(X-X_*)\left( \rho(X) - \rho(X_*)\right) \right|^2 \notag\\
&= 
\EE\left| \Pi(X-X_*)  \left( \rho(X) - \rho(X_*)\right)  \right|^2 - 4 \EE\left| \Pi(W) \rho_W(W)  \right|^2, 
\end{align}
so it would be sufficient to bound
\begin{align}
4 \EE\left| \Pi(W) \rho_W(W)  \right|^2 < \alpha \, \EE\left| \Pi(X-X_*)  \left( \rho(X) - \rho(X_*)\right)  \right|^2
\end{align}
for some $\alpha<1$ (possibly depending on the distribution of $X$) in order for the proof to carry over to general $X$.  Note that the distribution of $W = X-X_*$ is symmetric in general, but not radially symmetric, which may be useful toward establishing such a bound.

\section{Estimate of the Deficit in the Log Sobolev Inequality }\label{sec:LSI}

The problem of quantitatively bounding the deficit in the logarithmic Sobolev in equality has received considerable attention lately (e.g., \cite{fathi2014quantitative, bobkov2014bounds, indrei2013quantitative}), but the problem of obtaining a satisfactory estimate that is dimension-free remains open in general.  In this section, we apply our previous results to bound the deficit in the LSI for radially symmetric functions.   

For a random vector $X$ on $\mathbb{R}^d$, define the entropy power of $X$ as
\begin{align}
N(X) = \frac{1}{2\pi e} \exp\Big( {\tfrac{2}{d}h(X)}\Big).
\end{align}
The following entropy power inequality was  proved by the author in \cite{courtade2016strengthening}:
\begin{theorem}
Let $X,X_*,Z$ be independent random vectors on $\mathbb{R}^d$, with $Z$ being Gaussian. Then
\begin{align}
N(X+X_*+Z)N(Z) + N(X)N(X_*) \leq N(X+Z)N(X_*+Z).
\end{align}
\end{theorem}
Particularizing to the case where $Z\sim N(0,\mathrm{I})$, we have the immediate corollary:
\begin{align}
\frac{N(X+\sqrt{t}Z)N(X_*+\sqrt{t}Z) - N(X)N(X_*)}{t}\geq N(X+X_*+ \sqrt{t} Z)  \geq N(X+X_*). 
\end{align}
Letting $t\to 0$ and applying de Bruin's identity gives the inequality:
\begin{align}
d  N(X+X_*) \leq N(X)N(X_*)\left(J(X) + J(X_*) \right).
\end{align}
Supposing $X,X_*$ are identically distributed, we obtain the following improvement of Stam's inequality \cite{stam1959some}, which states that $N(X)J(X) \geq d$:
\begin{align}
\tfrac{1}{d}N(X) J(X) \geq \exp\left\{\frac{2}{d}\left( h(\tfrac{1}{\sqrt{2}}(X+X_*))-h(X)   \right) \right\} . \label{improvedStam}
\end{align}
By work of Carlen \cite{carlen1991superadditivity}, it is well-known that Stam's inequality is equivalent to Gross' log Sobolev inequality for the Gaussian measure \cite{gross1975logarithmic}; i.e., $\tfrac{1}{2}I(X|G) \geq D(X|G)$.  Using the inequality $\log x\leq x-1$, it is straightforward to convert \eqref{improvedStam} into the following improved log Sobolev inequality:
\begin{align}
\delta_{\mathsf{LSI}}(X)  \geq  \left(  h(\tfrac{1}{\sqrt{2}}(X+X_*))-h(X)  \right), \label{improvedLSI}
\end{align}
where  $\delta_{\mathsf{LSI}}(X) := \frac{1}{2}I(X|G) - D(X|G)$ denotes the deficit in the log Sobolev inequality.

In view of Corollary \ref{cor:HjumpNoRegularity}, we have established the following logarithmic Sobolev inequality in quantitative form for radially symmetric $X$:
\begin{theorem}\label{thm:stabilityNoReg}
Let $X$ be a radially symmetric random vector on $\mathbb{R}^d$,  $d\geq 2$, with $D(X)<\infty$. It holds that 
\begin{align}
\delta_{\mathsf{LSI}}(X) \geq  \sup_{\varepsilon>0 }\widetilde{C}_{\varepsilon}(X)
   \frac{ D(X)^{1+3\varepsilon}}{I(X)^{2\varepsilon}},
\end{align}
where $\widetilde{C}_{\varepsilon}(X)$ is defined as in \eqref{CtildeDefn}.
\end{theorem}

Theorem \ref{thm:stabilityNoReg} yields a quantitative stability result for the LSI whenever $\EE|X|^{p}<\infty$ for some $p>4$.  As an illustrative example, suppose $X$ is normalized such that $\EE|X|^2=d$ (so that $D(X)=D(X|G)$ and $I(X)=I(X|G)$) and $\EE|X|^8<\infty$.  Then,  a clean bound is  obtained by setting $\varepsilon=1/2$, rearranging, solving a quadratic inequality and bounding constant terms:
\begin{align}
\frac{1}{2} I(X) %
&\geq  D(X) \frac{1}{2}\left(1+
\sqrt{1+ \frac{1}{10^8}  \frac{  \,\, \||X|^2\|^{5/2}_1  }{    %
   \left\|  |X|^2 \right\|^2_{ 4}   }   \sqrt{ D(X)} }\right) .
\end{align}
Note that this is a very strong stability result. 
Indeed, it is a straightforward exercise (identify $\||X|^2\|_1 = d$, apply Minkowski's inequality to $\left\|  |X|^2 \right\|_{ 4}$ and simplify) to show that  
\begin{align}
D(X) \leq 10^8 \max\left\{ \delta_{\mathsf{LSI}}(X), \delta^{1/2}_{\mathsf{LSI}}(X)  \sqrt{\frac{\EE|X_1|^8}{d}}\right\}.
\end{align}
We emphasize that $\EE|X_1|^8$ is the eighth moment of the \emph{one-dimensional} random variable $X_1$, the first coordinate of $X$.  Hence,  we would generally expect that $\EE|X_1|^8 \ll d$   in high dimension.

If $X$ is known to be $c$-regular for some $c$, then we can establish the following  stability estimate, which has no explicit dependence on dimension:
\begin{theorem}\label{thmLSIdeficit}
Let $X$ be a radially symmetric random vector  on $\mathbb{R}^d$,  $d\geq 2$, with $c$-regular density $f$.  For any $\varepsilon > 0$
\begin{align}
\delta_{\mathsf{LSI}}(X)  &\geq   \frac{1}{4} K_{\varepsilon}(X)    I(X)^{1+\varepsilon}, 
\end{align}
where $K_{\varepsilon}(X)$ is as defined in \eqref{KpX_FI}.
\end{theorem}
\begin{proof}
All quantities in \eqref{improvedLSI} are invariant to translations of $X$, so we assume without loss of generality that $\EE X=0$ for the moment.  In this case, we may write 
\begin{align}
\delta_{\mathsf{LSI}}(X)  &\geq  D(X|G) - D( \tfrac{1}{\sqrt{2}}(X+X_*) |G) \\
&= -\delta_{\mathsf{LSI}}(X)  + \frac{1}{2}\left( I(X|G) - I( \tfrac{1}{\sqrt{2}}(X+X_*) |G)  \right)  + \delta_{\mathsf{LSI}}( \tfrac{1}{\sqrt{2}}(X+X_*))\\
&\geq   -\delta_{\mathsf{LSI}}(X)  + \frac{1}{2}\left( J(X) - J( \tfrac{1}{\sqrt{2}}(X+X_*) )  \right)  .
\end{align}
Thus, in view of Theorem \ref{thm:FIjumps}, we have proved the claim.
\end{proof}

\begin{remark}
Although Theorems \ref{thm:stabilityNoReg} and  \ref{thmLSIdeficit} consider radially symmetric densities, we see from \eqref{improvedLSI} and the following inequalities that $\delta_{\mathsf{LSI}}(X)$ dominates entropy production and dissipation of Fisher information  in general.  Stated another way, any quantitative lower bound on entropy production (or Fisher information dissipation) will provide a lower bound on $\delta_{\mathsf{LSI}}(X)$.  
\end{remark}

In closing, we mention here that a quantitative form of the sharp Sobolev inequality was established by Cianchi, Fusco,   Maggi and   Pratelli \cite{cianchi2009sharp} through a reduction of the general inequality to the setting of radially symmetric functions by considering spherically symmetric rearrangements. On this note, from \cite[Corollary 8.7]{wang2014beyond}, it is easy to verify that $\delta_{\mathsf{LSI}}(X)  \geq \delta_{\mathsf{LSI}}(X^{\star})$, where as before, $X^{\star}$ is the symmetric decreasing rearrangement of $X$.  Provided  $X^{\star}$ has regular density, we obtain 
\begin{align}
\delta_{\mathsf{LSI}}(X)  &\geq   \frac{1}{4} K_{\varepsilon}(X^{\star})    I(X^{\star} )^{1+\varepsilon}.
\end{align}

This should be compared to a result by Bobkov, Gozlan,   Roberto  and Samson \cite{bobkov2014bounds} (see also \cite{indrei2013quantitative}):   If $X$ is a random vector on $\mathbb{R}^d$ with smooth density $p = e^{-V}$ satisfying $V''\geq \varepsilon \mathrm{I}$ for some $\varepsilon>0$, then 
 \begin{align}
\delta_{\mathsf{LSI}}(X)  &\geq   c_{\varepsilon}    W_2^2(\bar{X},G), 
\end{align}
where $c_{\varepsilon}$ is a constant depending on $\varepsilon$, only, and  $\bar{X}$ corresponds to a rearrangement of the vector $X$ such that its one-dimensional marginals $\bar{X}_1, \bar{X}_2, \dots, \bar{X}_d$ form a martingale  \cite{bobkov2014bounds}. We remark that the log-concavity assumption can be more restrictive than our regularity assumptions, but the symmetry of $\bar{X}$ is less restrictive than radial symmetry.  We also mention that Fathi, Indrei and   Ledoux \cite{fathi2014quantitative} have established another quantitative estimate on $\delta_{\mathsf{LSI}}(X)$, under the assumption that $X$ satisfies a Poincar\'e inequality.   Namely, if $X$ is a centered random vector that has spectral gap $\varepsilon$, then 
 \begin{align}
\delta_{\mathsf{LSI}}(X)  &\geq   c_{\varepsilon}    I(X), 
\end{align}
where   $c_{\varepsilon}$ is a constant depending on $\varepsilon$, only.  It is interesting to recall from our above discussion that stability of entropy jumps (in dimension one) has previously been  established under a spectral gap condition \cite{ball2003entropy}.  We have now seen  in \eqref{improvedLSI} that the two stability problems are closely connected.

\section*{Appendix: Proof of Propositions \ref{prop:StartRegular} and \ref{prop:approxR}}

\begin{proof}[Proof of Proposition \ref{prop:StartRegular}]
First, suppose $X$ is $c$-regular, and let $V = X+Z$, $Z\sim N(0,\sigma^2 I)$.    We claim that $V$ is $(5c)$-regular. Toward this end, let $f_X$, $f_Z$ and $f_V$ denote the densities of $X,Z,V$, respectively. Now, 
\begin{align}
\left| \nabla f_V(v) \right| &=\left|   \int (\nabla f_X(v-z)) f_Z(z) dz \right|\\
&\leq   \int  \left| \nabla f_X(v-z)\right| f_Z(z) dz \\
&\leq c \int   f_X(v-z) \left(  \left|v-z\right| +\EE|X|  \right) f_Z(z) dz \\
&\leq c\, f_V(v) \left(|v| + \EE|X|  \right)   + c\, \int \left| z\right|  f_X(v-z)    f_Z(z) dz \\
&= c\, f_V(v) \left(|v| + \EE|X|  \right)   + c\, f_V(v)\, \EE\left[ \,|X-v| \, |V=v \right] \\
&\leq c\, f_V(v) \left(|v| + \EE|X|  \right)   + c\, f_V(v)\, \left( 3 |v| + 4 \EE|X| \right) \\
&\leq 5 c\, f_V(v) \left(|v| + \EE|V|  \right).
\end{align}
A proof of the only nontrivial  inequality $\EE\left[ \,|X-v| \, |V=v \right]  \leq 3 |v| + 4 \EE|X|$ can be found in \cite[Proposition 2]{polyWuWasserstein2016}. Since $V$ is smooth with nonvanishing density, $\nabla \log f_V  = \frac{\nabla f_V }{f_V}$; this completes the proof of the claim.

By a change of variables, we observe that  $e^{-t} X$ is $(c\,e^{2t})$-regular.  Combining with the previous claim finishes the proof. 
\end{proof}

\begin{remark}
The claim of Proposition \ref{prop:StartRegular} may be strengthened.  Indeed, if $X$ is $c$-regular, then $X_t$ is $(5c + 4 )$-regular for all $t\geq 0$.  To see this, note that Proposition \ref{prop:addNoise} establishes that $X_t$ is $(4 /
(1-e^{-2t}))$-regular.  However, Proposition \ref{prop:StartRegular} shows that $X_t$ is $(5c\,e^{2t})$-regular; maximizing the minimum of these two quantities over $t\geq 0$ establishes the strengthened claim.   The weaker claim  is more convenient for our purposes, so that regularity of $X_t$ is a multiple of $c$.  This ensures that the inequalities are invariant to scaling of $X$.
\end{remark}

\begin{proof}[Proof of Proposition \ref{prop:approxR}]
The proof is elementary, but provided here for completeness.  Define $X_0 = \sqrt{d}R_0 U$, where as before $U$ is uniform on $\mathbb{S}^{d-1}$ and independent of $R_0$.  Now, let 
\begin{align}
X = \sqrt{1-\varepsilon}X_{0} + \sqrt{\varepsilon} G,
\end{align}
and set $R = \frac{1}{\sqrt{d}} |X|$.  Clearly, $\EE R^2 = 1$ and, by Proposition \ref{prop:addNoise}, $X$ is $(4/\varepsilon)$-regular.  Now, observe that 
\begin{align}
F_{R}(r) &= \Pr\left\{\tfrac{1}{\sqrt{d}}  | \sqrt{1-\varepsilon}X_{0} + \sqrt{\varepsilon} G |  \leq  r   \right\}\\
&\geq \Pr\left\{  \tfrac{\sqrt{1-\varepsilon}}{\sqrt{d}}|X_{0}| + \tfrac{\sqrt{\varepsilon}}{\sqrt{d}}|G|   \leq  r   \right\}\\
&\geq \Pr\left\{ \tfrac{\sqrt{1-\varepsilon}}{\sqrt{d}} |X_{0}|   \leq  r -\sqrt{(t+1)\varepsilon}  \right\}\Pr\left\{   \tfrac{\sqrt{\varepsilon}}{\sqrt{d}}|G|   \leq  \sqrt{(t+1)\varepsilon}  \right\}\\
&= F_{R_0}\left(  \tfrac{ r -\sqrt{(t+1)\varepsilon}}{\sqrt{1-\varepsilon}}  \right) \left(1- \Pr\left\{   \tfrac{1}{{d}} |G|^2 -1  > t\right\}\right)\\
&\geq F_{R_0}\left(  \tfrac{ r - \sqrt{(t+1)\varepsilon}}{\sqrt{1-\varepsilon}}  \right)   - e^{-d t^2 /8} ,
\end{align}
where the final inequality is due to the tail bound for a chi-squared random variable with $d$ degrees of freedom.   The other direction is similar: 
\begin{align}
F_{R}(r) &= \Pr\left\{\tfrac{1}{\sqrt{d}}  | \sqrt{1-\varepsilon}X_{0} + \sqrt{\varepsilon} G |  \leq  r   \right\}\\
&\leq  \Pr\left\{  \tfrac{\sqrt{1-\varepsilon}}{\sqrt{d}}|X_{0}| - \tfrac{\sqrt{\varepsilon}}{\sqrt{d}}|G|   \leq  r   \right\}\\
&\leq \Pr\left\{ \tfrac{\sqrt{1-\varepsilon}}{\sqrt{d}} |X_{0}|   \leq  r +\sqrt{(t+1)\varepsilon}  \right\} + \Pr\left\{   \tfrac{\sqrt{\varepsilon}}{\sqrt{d}}|G|   >  \sqrt{(t+1)\varepsilon}  \right\}\\
&= F_{R_0}\left(  \tfrac{ r + \sqrt{(t+1)\varepsilon}}{\sqrt{1-\varepsilon}}  \right)  +  e^{-d t^2 /8}.
\end{align}
 \end{proof}
 
 
 \subsection*{Acknowledgment}
This work was supported in part by NSF grants CCF-1528132 and CCF-0939370 (Center for Science of Information).

\bibliographystyle{unsrt}
\bibliography{jumpsBib}

\newpage

\end{document}